\documentclass[a4paper]{article}

\title{Superlogarithmic Gap Result for LCLs on Trees in Quantum-LOCAL}

\usepackage{amsmath,amsthm,thmtools,amsfonts,amssymb}
\usepackage{mathrsfs}
\usepackage{mathtools}
\usepackage{physics}
\usepackage{braket}
\usepackage{array}
\allowdisplaybreaks

\usepackage{xspace}
\usepackage{xcolor}
\usepackage{subcaption}
\usepackage{enumitem}
\usepackage{microtype}
\usepackage{csquotes}
\usepackage{dsfont}
\usepackage{geometry}
\usepackage{babel}

\usepackage{algorithm}
\usepackage{algpseudocode}

\usepackage{cite}
\usepackage[hidelinks]{hyperref}
\usepackage{cleveref}
\crefname{lemma}{Lemma}{Lemmas}
\Crefname{lemma}{Lemma}{Lemmas}
\crefname{step}{step}{steps}
\Crefname{step}{Step}{Steps}

\declaretheorem[numberwithin=section]{theorem}
\declaretheorem{corollary, lemma, proposition}[sibling=theorem]
\declaretheorem[style=definition]{definition, example, remark}[sibling=theorem]

\usepackage[skins,many]{tcolorbox}

\newcommand{\calA}{\mathcal{A}}

\newcommand{\calC}{\mathcal{C}}

\newcommand{\calH}{\mathcal{H}}

\newcommand{\calN}{\mathcal{N}}

\newcommand{\calV}{\mathcal{V}}

\newcommand{\bbE}{\mathbb{E}}

\newcommand{\bbR}{\mathbb{R}}

\newcommand{\reals}{\bbR}

\DeclareMathOperator{\poly}{poly}
\DeclareMathOperator{\diam}{diam}

\DeclareMathOperator{\dist}{dist}
\DeclareMathOperator{\neighborhood}{\calN}
\DeclareMathOperator{\view}{\calV}
\DeclareMathOperator{\halfEdges}{\calH}
\DeclareMathOperator{\ecc}{ecc}
\DeclarePairedDelimiterXPP{\floor}[1]{}{\lfloor}{\rfloor}{}{#1}

\DeclarePairedDelimiterXPP{\expect}[1]{\bbE}{[}{]}{}{#1}
\DeclarePairedDelimiterXPP{\pr}[1]{\Pr}{[}{]}{}{#1}
\DeclarePairedDelimiterXPP{\variance}[1]{\mathrm{Var}}{[}{]}{}{#1}
\DeclarePairedDelimiterXPP{\covariance}[1]{\mathrm{Cov}}{[}{]}{}{#1}

\newcommand{\local}{LOCAL\xspace}

\newcommand{\outcome}{\textrm{Out}}
\newcommand{\inpt}{\textrm{inpt}}
\newcommand{\VR}{V^{\mathsf{R}}}
\newcommand{\VC}{V^{\mathsf{C}}}
\newcommand{\GR}{G^{\mathsf{R}}}
\newcommand{\GC}{G^{\mathsf{C}}}

\begin{document}

\newcommand{\myaff}[1]{\,$\cdot$\, {\small #1}\par\medskip}
\newenvironment{myabstract}%
{\list{}{\listparindent 1.5em
        \itemindent    \listparindent
        \leftmargin    0cm
        \rightmargin   0cm
        \parsep        0pt}%
    \item\relax}%
{\endlist}

\newenvironment{mycover}%
{\list{}{\listparindent 0pt
        \itemindent    \listparindent
        \leftmargin    0cm
        \rightmargin   1.5cm
        \parsep        0pt}%
    \raggedright
    \item\relax}%
{\endlist}

\begin{mycover}
{\huge\bfseries\boldmath Superlogarithmic Gap Result for LCLs on Trees in Quantum-LOCAL \par}
\bigskip
\bigskip

\textbf{Francesco d'Amore}
\myaff{Gran Sasso Science Institute, Italy}

\textbf{Henrik Lievonen}
\myaff{Aalto University, Espoo, Finland \,$\cdot$\, CISPA Helmholtz Center for Information Security, Saarbr\"ucken, Germany}

\bigskip
\end{mycover}

\begin{myabstract}
    \noindent\textbf{Abstract.}
    We show that, on trees, any locally checkable labeling problem (LCL) \(\Pi\) that can be solved by an \(n^{o(1)}\)-dependent distribution can also be solved by an \(O(\log n)\)-round deterministic LOCAL algorithm.
	The result is obtained through a rake-and-compress-style decomposition of the input tree, and local simulations of the bounded dependent distribution on the components of the decomposition.
	As a corollary to our result, any LCL problem on trees can either be solved by an \(O(\log n)\) deterministic LOCAL algorithm, or requires \(n^{\Omega(1)}\) rounds to solve by a quantum-LOCAL algorithm.
\end{myabstract}

\begin{myabstract}
    \noindent\textbf{Acknowledgements.}
    Francesco d’Amore was supported by Decreto MUR n. 47/2025, CUP D13C25000750001.
    This work was supported in part by the Research Council of Finland, Grants 363558. Most of this work was done while Henrik Lievonen was affiliated with Aalto University.
\end{myabstract}

% \thispagestyle{empty}
% \setcounter{page}{0}
% \clearpage

\section{Introduction}
The distributed computing community has recently made significant progress in understanding the role of quantum computation and communication in the context of distributed graph problems \cite{le-gall-2022-quantum-distributed-computing,le-gall-magniez-2018-sublinear-time-quantum-computation,izumi-le-gall-2019-quantum-distributed-algorithm-for,censor-hillel-fischer-etal-2022-quantum-distributed}, even with respect to locally checkable labeling (LCL) problems in the LOCAL model \cite{coiteux-roy-d-amore-etal-2024-no-distributed-quantum,akbari-coiteux-roy-etal-2025-online-locality-meets,balliu-brandt-etal-2025-distributed-quantum-advantage,d-amore-2025-on-the-limits-of-distributed-quantum}.
LCL problems are problems whose solutions are \emph{easy to verify}, but might be \emph{hard to find} by a distributed algorithm.
More precisely, an LCL problem can be specified by a finite list \(\calC\) of \emph{valid labeled neighborhoods} of constant radius and of bounded degree, and a node can verify its local solution by simply inspecting is constant-radius neighborhood and checking whether it belongs to \(\calC\) \cite{naor-stockmeyer-1995-what-can-be-computed-locally}.
Many classical distributed graph problems, such as vertex coloring, maximal independent set, maximal matching, sinkless orientation, etc., belong to this class:
in essence, the class of LCL problems is wide enough to include many interesting natural graph problems but restrictive enough so that many results apply to the class as a whole --
in fact, the complexity of LCL problems in the classical LOCAL model is nowadays well-understood \cite{naor-stockmeyer-1995-what-can-be-computed-locally,chang-pettie-2019-a-time-hierarchy-theorem-for-the,chang-kopelowitz-pettie-2019-an-exponential-separation,chang-2020-the-complexity-landscape-of-distributed,brandt-fischer-etal-2016-a-lower-bound-for-the,balliu-hirvonen-etal-2018-new-classes-of-distributed,brandt-2019-an-automatic-speedup-theorem-for,balliu-brandt-etal-2021-lower-bounds-for-maximal,balliu-brandt-etal-2021-almost-global-problems-in-the,suomela-2020-landscape-of-locality-invited-talk,dahal-d-amore-etal-2023-brief-announcement-distributed}, and we especially know when randomness helps (or does not help) in solving them\cite{chang-kopelowitz-pettie-2019-an-exponential-separation,chang-pettie-2019-a-time-hierarchy-theorem-for-the,balliu-brandt-etal-2020-how-much-does-randomness-help,balliu-ghaffari-etal-2025-shared-randomness-helps-with}.
While general quantum advantage in the LOCAL model was already established in 2019 \cite{le-gall-nishimura-rosmanis-2019-quantum-advantage-for}, the complexity landscape of LCL problems in the quantum-LOCAL model \cite{gavoille-kosowski-markiewicz-2009-what-can-be-observed} is much less understood, with mostly negative results showing that quantum computation and communication cannot bring (much) advantage \cite{gavoille-kosowski-markiewicz-2009-what-can-be-observed,arfaoui-fraigniaud-2014-what-can-be-computed-without,coiteux-roy-d-amore-etal-2024-no-distributed-quantum,akbari-coiteux-roy-etal-2025-online-locality-meets,d-amore-2025-on-the-limits-of-distributed-quantum,dhar-kujawa-etal-2024-local-problems-in-trees-across-a,balliu-coupette-etal-2025-new-limits-on-distributed}.
Only few positive results have been recently provided \cite{balliu-brandt-etal-2025-distributed-quantum-advantage,balliu-casagrande-etal-2026-distributed-quantum}, exhibiting artificially crafted LCL problems that can be solved faster in quantum-LOCAL than in classical LOCAL, which leave open whether natural LCL problems admit quantum advantage.

On the negative side, it has been proved that for approximate graph coloring problems, quantum-\local algorithms cannot be substantially faster than deterministic \local algorithms \cite{coiteux-roy-d-amore-etal-2024-no-distributed-quantum}, and that linear programming problems do not admit any quantum advantage with respect to deterministic algorithms \cite{balliu-coupette-etal-2025-new-limits-on-distributed}.
More generally, there have been works showing general negative results for LCL problems on specific graph families:
It is known that on rooted trees, \(o(\log \log \log n)\)-round quantum algorithms for an LCL can be transformed into \(O(\log^\star n)\)-round deterministic \local algorithms \cite{akbari-coiteux-roy-etal-2025-online-locality-meets}.
Moreover, Dhar et al.\ \cite{dhar-kujawa-etal-2024-local-problems-in-trees-across-a} proved the following statements:
\begin{enumerate}
	\item On rooted regular trees, there is no possible quantum advantage unless the classical deterministic complexity is \(O(\log^\star n)\).
	\item On rooted trees, a global problem (i.e., a problem with classical locality \(\Omega(n)\)) cannot be solved in quantum-LOCAL with locality \(o(n)\).
	\item On unrooted regular trees, if the deterministic complexity is \(\omega(\log n)\), then there cannot be any quantum advantage.
	\item On unrooted trees, a global problem cannot be solved in quantum-LOCAL with locality \(o(n)\).
\end{enumerate}
The situation in trees is quite unsatisfactory, especially because most of the classification results for LCL problems in the classical LOCAL model are specific to the family of trees \cite{chang-2020-the-complexity-landscape-of-distributed}.
In this work, we make a step towards a better comprehension of the quantum complexity landscape on trees.
We prove that, on unrooted trees, any LCL problem that can be solved by an \(n^{o(1)}\)-round quantum-LOCAL algorithm can also be solved by an \(O(\log n)\)-round deterministic LOCAL algorithm. 
The same result trivially extends to the case of rooted trees, which was left open by \cite{dhar-kujawa-etal-2024-local-problems-in-trees-across-a}.

We actually prove a stronger result: on unrooted trees, any LCL problem that can be solved by an \(n^{o(1)}\)-dependent distribution can also be solved by an \(O(\log n)\)-round deterministic LOCAL algorithm.
We recall that a \(T\)-dependent distribution is a probability distribution over the output labelings of the graph such that, for every two subsets of nodes \(A\) and \(B\) at distance at least \(T+1\), the random variables corresponding to the output labels of nodes in \(A\) and those in \(B\) are mutually independent.
Note that a \(T\)-round quantum-LOCAL algorithm always produces a \(2T\)-dependent distribution \cite{akbari-coiteux-roy-etal-2025-online-locality-meets}: hence, negative results for \(T\)-dependent distributions also apply to quantum-LOCAL algorithms with locality \(\floor{T/2}\).

\section{Preliminary definitions}

We present here informal definitions of the models of computation we use.
We refer the reader to Appendix~\ref{sec:preliminaries} for more formal definitions.

\subsection{The models}

We work in the LOCAL model of distributed computing: in this model, the computation proceeds in synchronous rounds, and in each round, each node can perform unlimited local computation and exchange messages of unbounded size with its neighbors.
Nodes receive as input the size \(n\) of the graph, and they have unique \(O(\log n)\)-bit identifiers (IDs) that are distinct from each other.
In the \emph{randomized} LOCAL model, as opposed to the \emph{deterministic} LOCAL model, nodes also have access to private source of randomness, and the algorithm is allowed to fail with probability at most \(1/n^c\) for any given \(c > 0\).
In the \emph{quantum}-LOCAL model, nodes can exchange an arbitrary number of qubits and perform arbitrary unitary transformations on their local quantum states: the output of a node is just the measurement outcome of its local quantum state, and the algorithm is allowed to fail with probability at most \(1/n^c\) for any given \(c > 0\).

Finally, we also consider the \emph{bounded-dependence} model \cite{akbari-coiteux-roy-etal-2025-online-locality-meets}, where we assign to each input graph a probability distribution over the output labelings, and we say that this assignment has locality \(T\) if the two following statements hold:
\begin{enumerate}
    \item \emph{Independence property}: For every input graph \(G\), and any two subsets of nodes \(A\) and \(B\) at distance at least \(2T+1\), the random variables corresponding to the output labels of nodes in \(A\) and those in \(B\) are mutually independent.
    In such case, we say that the distribution is \(2T\)-dependent.
    \item \emph{Non-signaling property}: For every two input graphs \(G\) and \(G'\) that are identical in the radius-\(T\) neighborhood of a subset of nodes \(A\), the marginal distributions of the output labels of \(A\) are the same in both distributions.
\end{enumerate}
Again, in this model we require that the distribution solves the problem with probability at least \(1 - 1/n^c\) for any given \(c > 0\).
It holds that a \(T\)-round LOCAL algorithm (being it deterministic, randomized, or quantum) always produces a bounded-dependent distribution with locality \(T\) \cite{akbari-coiteux-roy-etal-2025-online-locality-meets}.
Hence, a lower bound result for \(2T\)-dependent distributions also applies to \(T\)-round quantum-LOCAL algorithms.

Without loss of generality, we assume that the output labels in all of these model are given on half-edges, that is, pairs of the form \((v,e)\) where \(v\) is a node and \(e\) is an edge incident to \(v\).

\subsection{Locally checkable labeling problems}

In this paper, we consider the family of locally checkable labeling (LCL) problems, introduced in \cite{naor-stockmeyer-1995-what-can-be-computed-locally}.
We here focus on a simpler---yet equivalent---definition of LCL problems, which is specific to trees: the family of node-edge checkable problems \cite{balliu-censor-hillel-etal-2021-locally-checkable}.
In this definition, an LCL problem is specified by a finite set of node constraints \(\calC_N\) and a finite set of half-edge constraints \(\calC_E\), where a node configuration is a multiset of input and output labels that can appear on the half-edges incident to a node, and a half-edge configuration is a pair of input and output labels that can appear on the two half-edges corresponding to an edge.
Since \(\calC_N\) is finite, it follows that the maximum degree of the input graph, which we denote by \(\Delta\), is a constant that is independent of the number of nodes \(n\).

\section{Overview of the analysis}

Our main technical contribution is a procedure that, given a bounded-dependent distribution with locality \(T = n^{o(1)}\) for a node-edge checkable problem \(\Pi\), produces a randomized LOCAL algorithm with locality \(O(\poly(T \log n))\) that solves \(\Pi\).
By known gap results \cite{chang-pettie-2019-a-time-hierarchy-theorem-for-the}, since \(O(\poly(T \log n)) = n^{o(1)}\), it holds that there exists a deterministic LOCAL algorithm with locality \(O(\log n)\) that solves~\(\Pi\).

Our randomized algorithm has two stages: 
We first perform a rake-and-compress type decomposition of the input tree, which allows us to divide the input graph into small components that follow a hierarchical structure.
Then, we sample the output labels from the bounded-dependent distribution for each component of the decomposition, starting from the highest layer in the hierarchy and proceeding downwards. 
The sampling procedure is not trivial, as we need to ensure that the output labels sampled for each component are consistent with those sampled for the upper layers of the hierarchy (since the algorithm does not have global view), and that the overall sampling procedure works with high probability.
The issue here is that the bounded-dependent distribution has dependencies at distances at most \(2T\), and when two components meet in downward layers, we need to make sure that the two distributions are compatible with each other.

\subsection{Hierarchical decomposition of a tree}

The hierarchical decomposition that we use is obtained through a rake-and-compress style decomposition.
Rake-and-compress has been extensively used in the literature \cite{chang-2020-the-complexity-landscape-of-distributed,miller-reif-1985-parallel-tree-contraction-and-its,chang-he-etal-2018-the-complexity-of-distributed-edge,chang-pettie-2019-a-time-hierarchy-theorem-for-the,balliu-brandt-etal-2023-on-the-node-averaged-complexity}.
It works roughly as follows:
Given a forest $G = (V, E)$, the nodes of the forest are decomposed into a sequence of pairwise disjoint sets \
\[ 
    V = \VR_{1} \cup \VC_{1} \cup \VR_{2} \cup \VC_{2} \cup \VR_{3} \cup \VC_{3} \cup \ldots
\] 
where the subscript and superscript together indicate the \emph{level} of the decomposition, and the superscript indicates the \emph{type} of the level, which is either \emph{rake} (denoted by \(\mathsf{R}\)) or \emph{compress} (denoted by \(\mathsf{C}\)).
We denote the forest induced by nodes $\VR_i \cup \VC_i \cup \VR_{i+1} \cup \VC_{i+1} \cup \cdots$ by $\GR_i$, and the forest induced by nodes $\VC_i \cup \VR_{i+1} \cup \VC_{i+1} \cup \cdots$ by $\GC_i$.
Let $L$ be the smallest integer \(i\) such that $\GC_i$ is the empty graph.

Given two parameters $\ell \ge 1$ and $\gamma \ge 1$, the decomposition satisfies the following two key properties:
\begin{enumerate}
	\item\label{prop:tree-decomposition-rake} Each connected component of the subgraph induced by $\VR_i$ must be a rooted tree of height at most $\gamma-1$ such that the root has at most one neighbor in $\GC_i$, while other nodes have no neighbors in $\GC_i$.
	\item\label{prop:tree-decomposition-compress} Each connected component of the subgraph induced by $\VC_i$ must be a path on a number of nodes in the range $[\ell, 2\ell]$ such that the endpoints have at most one neighbor in $\GR_{i+1}$, while intermediate nodes have no neighbors in \(\GR_{i+1}\). If the path has one node, then this single node can have at most two neighbors in \(\GR_{i+1}\).
\end{enumerate}

A decomposition satisfying these properties is called a \emph{$(\gamma, \ell)$-decomposition}.
We order the layers as
\(\VR_1,\VC_1,\VR_2,\VC_2,\ldots\).  For later use, for each connected
component \(C\) of a layer, we define the \emph{upward-component of \(C\)},
denoted \(C^+\), as the union of all components reachable from \(C\) by a path
that visits components in strictly increasing layer order.  Equivalently,
starting from \(C\), we repeatedly add neighboring components that belong to
higher layers.
While many algorithms for this decomposition have been already given in previous works, we need an algorithm that takes care of the case in which both \(\gamma\) and \(\ell\) are non-constant parameters.
In the appendix, we prove the following theorem.
\begin{theorem}
    \label{thm:rcdecomp}
    Let \(\gamma,\ell\ge 1\) be integers.
    There exists a deterministic \local algorithm \(\calA\) that computes a
    \((\gamma,\ell)\)-decomposition of a forest of \(n\) nodes with
    \(L = O(\lceil \ell/\gamma\rceil\log n)\) levels in time
    \(O((\gamma+\ell^2+\ell\log^* n)\lceil \ell/\gamma\rceil\log n)\).
\end{theorem}

A key property of any \((\gamma,\ell)\)-decomposition is the following lemma, whose proof we leave in the appendix.

\begin{restatable}{lemma}{lemmadist}
    \label{lem:dist}
    Let \(\gamma,\ell\ge 1\) be integers, let \(G\) be a forest, and let
    \[ 
        V = \VR_{1} \cup \VC_{1} \cup \VR_{2} \cup \VC_{2} \cup \VR_{3} \cup \VC_{3} \cup \ldots
    \] 
    be a \((\gamma, \ell)\)-decomposition of \(G\).
    Let \(C\) be a connected component in one of the layers \(\VR_j\) or \(\VC_j\).
    Then for every layer \(V_i\) (either \(\VR_i\) or \(\VC_i\)), it holds that any two components \(C_1\) and \(C_2\) of \(V_i \cap C^+\) are at distance at least \(\ell+1\) from each other.
\end{restatable}

\section{Sampling procedure}
In this section, we present a \local algorithm that uses a bounded-dependence distribution to produce a solution for an LCL problem \(\Pi\).
Formally, we prove the following lemma which, combined with \cref{thm:rcdecomp}, gives the full algorithm:
\begin{lemma}
    \label{lem:sampling}
	Let $\Pi$ be an LCL problem on forests, and let $\outcome$ be a bounded-dependence distribution with locality $T(n)$ that works with high probability.
	Then, there exists a \local algorithm that, given a $(T(n), 2T(n))$-decomposition of the input graph $G$ with size $n$, solves problem $\Pi$ with high probability and locality $O(T(n) \log n)$.
\end{lemma}

\subsection{Sampling algorithm}

Let $n$ be the size of the input instance, and let $c > 0$ be the constant of ``with high probability''.
We fix the probability constant\footnotemark{} of the bounded-dependence distribution $\outcome$ to $c + 3$, we fix its locality $T = T(n)$.
Let
\[
    V = \VR_1 \cup \VC_1 \cup \VR_2 \cup \VC_2 \cup \cdots \cup \VR_L
\]
be a $(T, 2T)$-decomposition, which can be computed using \cref{thm:rcdecomp}.
\footnotetext{
    Some authors mean by ``with high probability'' the case of \(c = 1\).
    In this case, instead of picking the constant of ``with high probability'', we lie to the underlying distribution that the graph has \(n^3\) nodes; this gives us the needed boost in success probability while still keeping the locality of the underlying distribution \(n^{o(1)}\).
}

The randomized \local algorithm works by labeling the nodes according to this decomposition, starting from $\VR_L$ and proceeding downwards.
Each layer of the decomposition consists of multiple components which the algorithm handles in parallel:
for each component $C$ of the layer currently being processed, the algorithm samples $\outcome$ conditioned on what has been sampled in its neighboring higher components $C^+ \setminus C$.

\subsection{Analysis}

\begin{proof}[Proof of \cref{lem:sampling}]
	We first start by analyzing the locality of the above \local algorithm, and then we analyze its correctness.
    For the locality, note that each connected component induced by the nodes of any given layer have diameter at most \(O(T)\); this is because rake components have radius \(T\) by construction, and compress components have length at most \(4T\) by construction.
    The sampling also needs access to the previously-sampled values of the neighboring higher components \(C^+ \setminus C\); we handle this by carrying the values downwards with the sampling:
    when a component is handled, it can read the sampled values of the higher-level components from its neighboring components, do the local sampling, and propagate the previous sample and the new sample to all nodes of the component -- all with locality \(O(T)\).
    Therefore, the sampling within each component can be done with locality \(O(T)\).
    As there are $L = O(\log n)$ layers, the algorithm takes in total $O(T \log n)$ locality.
    
    It remains to prove that the algorithm works correctly with probability at least $1 - 1/n^c$:
    we prove this using the union bound over local failure probabilities.
    Intuitively, if the algorithm produces an error somewhere with too large of a probability, then there exists a sampling order that produces genuine sample from the original distribution with the same failure probability.
    
    Fix a node \(v\), and consider the probability \(p_v\) that the \local algorithm fails to label the node correctly, that is, the label of \(v\) is not in \(\calC_N\).
    Let \(V_i\) (either \(\VR_i\) or \(\VC_i\)) be the layer in which node \(v\) resides, and let \(C\) be the corresponding connected component of \(V_i\).
    Let \(C^+\) be the upward-component of \(C\).
    Consider now the \emph{global} sampling process where we sample the nodes of
    \(C^+\) from the highest layer downwards: first
    \(\VR_L \cap C^+\), then \(\VC_{L-1}\cap C^+\) conditioned on the output of
    \(\VR_L \cap C^+\), then \(\VR_{L-1}\cap C^+\) conditioned on both higher
    layers, and so on, until we reach layer \(V_i\).
    Note that this produces a genuine sample from distribution \(\outcome\) for all nodes of \(C^+\).

    The key observation is that this produced distribution on \(C^+\) matches exactly the distribution that the \local algorithm produces.\footnotemark
    \footnotetext{However, this does not necessarily hold for all nodes of the graph.}
    This is because the \local algorithm carries these sampling dependencies down.
    A rake component has at most one neighboring component in higher layers.  A
    compress component may have two such neighboring components, but the two
    upward-components rooted at them are separated by the compress path, which
    contains at least \(2T\) nodes.  Hence the two upward-components are at
    distance at least \(2T+1\), and their samples are independent in the
    original distribution \(\outcome\) by \cref{lem:dist}.  Thus sampling the
    two sides separately, and then sampling the compress component conditioned
    on both sides, gives the same distribution as the global top-down sampling.
    Therefore, the output distribution of node \(v\) must match that of the original distribution \(\outcome\), and in particular, it cannot fail with probability higher than \(1/n^{c+3}\), as otherwise the original distribution would fail with at least this probability.
    
    Using similar analysis, we can also bound the probability that the output of an edge violates \(\calC_E\) is at most \(1/n^{c+3}\).
    
    Now, for the \local algorithm to fail, it must fail locally at a node or at an edge.
    As the graph has \(n\) nodes and at most \(n-1\) edges, the total failure probability can be bounded by the union bound:
    \[
        \Pr\Biggl[\Bigl(\bigcup_{v \in V(G)} \text{node $v$ fails} \Bigr) \cup \Bigl(\bigcup_{e \in E(G)} \text{edge $e$ fails} \Bigr)\Biggr]
        \le n \cdot \frac{1}{n^{c+3}} + (n-1) \cdot \frac{1}{n^{c+3}}
        \le \frac{1}{n^c}
    \]
    Therefore, the \local algorithm works with probability at least \(1-1/n^c\), completing the proof.
\end{proof}

We can now proof the main theorem of this paper:
\begin{theorem}
	Let \(\outcome\) be a bounded-dependence distribution solving an LCL problem \(\Pi\) on forests with locality \(n^{o(1)}\) and with high probability.
	Then there exists a deterministic \local algorithm that solves problem \(\Pi\) with locality \(O(\log n)\).
\end{theorem}
\begin{proof}
    Let \(T(n) = n^{o(1)}\) be the locality of the bounded-dependence distribution \(\outcome\).
	Combining \cref{thm:rcdecomp} and \cref{lem:sampling}, we get a randomized \local algorithm that solves \(\Pi\) with locality
	\[
	    O\bigl( (T^2(n) + T(n)\log^*n) \log n \bigr)
		= n^{o(1)} .
	\]
	We combine this with an existing speed-up result on LCL problems in trees to get an algorithm with locality \(O(\log n)\) \cite{chang-pettie-2019-a-time-hierarchy-theorem-for-the}.
\end{proof}

\clearpage
\bibliographystyle{plain}
\bibliography{biblio.bib}

@inproceedings{akbari-coiteux-roy-etal-2025-online-locality-meets,
  author = {Amirreza Akbari and Xavier Coiteux-Roy and Francesco D'Amore and Fran{\c{c}}ois {Le Gall} and Henrik Lievonen and Darya Melnyk and Augusto Modanese and Shreyas Pai and Marc-Olivier Renou and V{\'a}clav Rozhon and Jukka Suomela},
  editor = {Michal Kouck{\'y} and Nikhil Bansal},
  title = {Online Locality Meets Distributed Quantum Computing},
  booktitle = {Proceedings of the 57th Annual {ACM} Symposium on Theory of Computing, {STOC} 2025, Prague, Czechia, June 23-27, 2025},
  pages = {1295--1306},
  publisher = {{ACM}},
  year = {2025},
  doi = {10.1145/3717823.3718211}
}

@article{arfaoui-fraigniaud-2014-what-can-be-computed-without,
  author = {Heger Arfaoui and Pierre Fraigniaud},
  title = {What can be computed without communications?},
  journal = {SIGACT News},
  volume = {45},
  number = {3},
  pages = {82--104},
  year = {2014},
  doi = {10.1145/2670418.2670440}
}

@inproceedings{balliu-brandt-etal-2020-how-much-does-randomness-help,
  author = {Alkida Balliu and Sebastian Brandt and Dennis Olivetti and Jukka Suomela},
  editor = {Yuval Emek and Christian Cachin},
  title = {How much does randomness help with locally checkable problems?},
  booktitle = {{PODC} '20: {ACM} Symposium on Principles of Distributed Computing, Virtual Event, Italy, August 3-7, 2020},
  pages = {299--308},
  publisher = {{ACM}},
  year = {2020},
  doi = {10.1145/3382734.3405715}
}

@article{balliu-brandt-etal-2021-almost-global-problems-in-the,
  author = {Alkida Balliu and Sebastian Brandt and Dennis Olivetti and Jukka Suomela},
  title = {Almost global problems in the {LOCAL} model},
  journal = {Distributed Computing},
  volume = {34},
  number = {4},
  pages = {259--281},
  year = {2021},
  doi = {10.1007/S00446-020-00375-2}
}

@article{balliu-brandt-etal-2021-lower-bounds-for-maximal,
  author = {Alkida Balliu and Sebastian Brandt and Juho Hirvonen and Dennis Olivetti and Mika{\"e}l Rabie and Jukka Suomela},
  title = {Lower Bounds for Maximal Matchings and Maximal Independent Sets},
  journal = {Journal of the ACM},
  volume = {68},
  number = {5},
  pages = {39:1--39:30},
  year = {2021},
  doi = {10.1145/3461458}
}

@inproceedings{balliu-brandt-etal-2023-on-the-node-averaged-complexity,
  author = {Alkida Balliu and Sebastian Brandt and Fabian Kuhn and Dennis Olivetti and Gustav Schmid},
  editor = {Rotem Oshman},
  title = {On the Node-Averaged Complexity of Locally Checkable Problems on Trees},
  booktitle = {37th International Symposium on Distributed Computing, {DISC} 2023, October 10-12, 2023, L'Aquila, Italy},
  series = {LIPIcs},
  volume = {281},
  pages = {7:1--7:21},
  publisher = {Schloss Dagstuhl - Leibniz-Zentrum f{\"{u}}r Informatik},
  year = {2023},
  doi = {10.4230/LIPICS.DISC.2023.7}
}

@inproceedings{balliu-brandt-etal-2025-distributed-quantum-advantage,
  author = {Alkida Balliu and Sebastian Brandt and Xavier Coiteux-Roy and Francesco D'Amore and Massimo Equi and Fran{\c{c}}ois {Le Gall} and Henrik Lievonen and Augusto Modanese and Dennis Olivetti and Marc-Olivier Renou and Jukka Suomela and Lucas Tendick and Isadora Veeren},
  editor = {Michal Kouck{\'y} and Nikhil Bansal},
  title = {Distributed Quantum Advantage for Local Problems},
  booktitle = {Proceedings of the 57th Annual {ACM} Symposium on Theory of Computing, {STOC} 2025, Prague, Czechia, June 23-27, 2025},
  pages = {451--462},
  publisher = {{ACM}},
  year = {2025},
  doi = {10.1145/3717823.3718233}
}

@inproceedings{balliu-casagrande-etal-2026-distributed-quantum,
  author       = {Alkida Balliu and
                  Filippo Casagrande and
                  Francesco d'Amore and
                  Massimo Equi and
                  Barbara Keller and
                  Henrik Lievonen and
                  Dennis Olivetti and
                  Gustav Schmid and
                  Jukka Suomela},
  editor       = {Kasper Green Larsen and
                  Barna Saha},
  title        = {Distributed Quantum Advantage in Locally Checkable Labeling Problems},
  booktitle    = {Proceedings of the 2026 Annual {ACM-SIAM} Symposium on Discrete Algorithms,
                  {SODA} 2026, Vancouver, BC, Canada, January 11-14, 2026},
  pages        = {1268--1308},
  publisher    = {{SIAM}},
  year         = {2026},
  url          = {https://doi.org/10.1137/1.9781611978971.49},
  doi          = {10.1137/1.9781611978971.49},
  bibsource    = {dblp computer science bibliography, https://dblp.org}
}

@inproceedings{balliu-censor-hillel-etal-2021-locally-checkable,
  author = {Alkida Balliu and Keren Censor-Hillel and Yannic Maus and Dennis Olivetti and Jukka Suomela},
  editor = {Seth Gilbert},
  title = {Locally Checkable Labelings with Small Messages},
  booktitle = {35th International Symposium on Distributed Computing, {DISC} 2021, October 4-8, 2021, Freiburg, Germany (Virtual Conference)},
  series = {LIPIcs},
  volume = {209},
  pages = {8:1--8:18},
  publisher = {Schloss Dagstuhl - Leibniz-Zentrum f{\"{u}}r Informatik},
  year = {2021},
  doi = {10.4230/LIPICS.DISC.2021.8}
}

@inproceedings{balliu-coupette-etal-2025-new-limits-on-distributed,
  author = {Alkida Balliu and Corinna Coupette and Antonio Cruciani and Francesco D'Amore and Massimo Equi and Henrik Lievonen and Augusto Modanese and Dennis Olivetti and Jukka Suomela},
  editor = {Dariusz R. Kowalski},
  title = {New Limits on Distributed Quantum Advantage: Dequantizing Linear Programs},
  booktitle = {39th International Symposium on Distributed Computing, {DISC} 2025, October 27-31, 2025, Berlin, Germany},
  series = {LIPIcs},
  volume = {356},
  pages = {11:1--11:22},
  publisher = {Schloss Dagstuhl - Leibniz-Zentrum f{\"{u}}r Informatik},
  year = {2025},
  doi = {10.4230/LIPICS.DISC.2025.11}
}

@inproceedings{balliu-ghaffari-etal-2025-shared-randomness-helps-with,
  author = {Alkida Balliu and Mohsen Ghaffari and Fabian Kuhn and Augusto Modanese and Dennis Olivetti and Mika{\"e}l Rabie and Jukka Suomela and Jara Uitto},
  editor = {Keren Censor-Hillel and Fabrizio Grandoni and Jo{\"e}l Ouaknine and Gabriele Puppis},
  title = {Shared Randomness Helps with Local Distributed Problems},
  booktitle = {52nd International Colloquium on Automata, Languages, and Programming, {ICALP} 2025, July 8-11, 2025, Aarhus, Denmark},
  series = {LIPIcs},
  volume = {334},
  pages = {16:1--16:18},
  publisher = {Schloss Dagstuhl - Leibniz-Zentrum f{\"{u}}r Informatik},
  year = {2025},
  doi = {10.4230/LIPICS.ICALP.2025.16}
}

@inproceedings{balliu-hirvonen-etal-2018-new-classes-of-distributed,
  author = {Alkida Balliu and Juho Hirvonen and Janne H. Korhonen and Tuomo Lempi{\"a}inen and Dennis Olivetti and Jukka Suomela},
  editor = {Ilias Diakonikolas and David Kempe and Monika Henzinger},
  title = {New classes of distributed time complexity},
  booktitle = {Proceedings of the 50th Annual {ACM} {SIGACT} Symposium on Theory of Computing, {STOC} 2018, Los Angeles, CA, USA, June 25-29, 2018},
  pages = {1307--1318},
  publisher = {{ACM}},
  year = {2018},
  doi = {10.1145/3188745.3188860}
}

@article{barenboim-elkin-goldenberg-2022-locally-iterative,
  author = {Leonid Barenboim and Michael Elkin and Uri Goldenberg},
  title = {Locally-iterative Distributed {$(\Delta+1)$}-coloring and Applications},
  journal = {Journal of the ACM},
  volume = {69},
  number = {1},
  pages = {5:1--5:26},
  year = {2022},
  doi = {10.1145/3486625}
}

@inproceedings{brandt-2019-an-automatic-speedup-theorem-for,
  author = {Sebastian Brandt},
  editor = {Peter Robinson and Faith Ellen},
  title = {An Automatic Speedup Theorem for Distributed Problems},
  booktitle = {Proceedings of the 2019 {ACM} Symposium on Principles of Distributed Computing, {PODC} 2019, Toronto, ON, Canada, July 29 - August 2, 2019},
  pages = {379--388},
  publisher = {{ACM}},
  year = {2019},
  doi = {10.1145/3293611.3331611}
}

@inproceedings{brandt-fischer-etal-2016-a-lower-bound-for-the,
  author = {Sebastian Brandt and Orr Fischer and Juho Hirvonen and Barbara Keller and Tuomo Lempi{\"a}inen and Joel Rybicki and Jukka Suomela and Jara Uitto},
  editor = {Daniel Wichs and Yishay Mansour},
  title = {A lower bound for the distributed {L}ov{\'a}sz local lemma},
  booktitle = {Proceedings of the 48th Annual {ACM} {SIGACT} Symposium on Theory of Computing, {STOC} 2016, Cambridge, MA, USA, June 18-21, 2016},
  pages = {479--488},
  publisher = {{ACM}},
  year = {2016},
  doi = {10.1145/2897518.2897570}
}

@inproceedings{censor-hillel-fischer-etal-2022-quantum-distributed,
  author = {Keren Censor-Hillel and Orr Fischer and Fran{\c{c}}ois {Le Gall} and Dean Leitersdorf and Rotem Oshman},
  editor = {Mark Braverman},
  title = {Quantum Distributed Algorithms for Detection of Cliques},
  booktitle = {13th Innovations in Theoretical Computer Science Conference, {ITCS} 2022, January 31 - February 3, 2022, Berkeley, CA, {USA}},
  series = {LIPIcs},
  volume = {215},
  pages = {35:1--35:25},
  publisher = {Schloss Dagstuhl - Leibniz-Zentrum f{\"{u}}r Informatik},
  year = {2022},
  doi = {10.4230/LIPICS.ITCS.2022.35}
}

@inproceedings{chang-2020-the-complexity-landscape-of-distributed,
  author = {Yi-Jun Chang},
  editor = {Hagit Attiya},
  title = {The Complexity Landscape of Distributed Locally Checkable Problems on Trees},
  booktitle = {34th International Symposium on Distributed Computing, {DISC} 2020, October 12-16, 2020, Virtual Conference},
  series = {LIPIcs},
  volume = {179},
  pages = {18:1--18:17},
  publisher = {Schloss Dagstuhl - Leibniz-Zentrum f{\"{u}}r Informatik},
  year = {2020},
  doi = {10.4230/LIPICS.DISC.2020.18}
}

@inproceedings{chang-he-etal-2018-the-complexity-of-distributed-edge,
  author = {Yi-Jun Chang and Qizheng He and Wenzheng Li and Seth Pettie and Jara Uitto},
  editor = {Artur Czumaj},
  title = {The Complexity of Distributed Edge Coloring with Small Palettes},
  booktitle = {Proceedings of the Twenty-Ninth Annual {ACM-SIAM} Symposium on Discrete Algorithms, {SODA} 2018, New Orleans, LA, USA, January 7-10, 2018},
  pages = {2633--2652},
  publisher = {{SIAM}},
  year = {2018},
  doi = {10.1137/1.9781611975031.168}
}

@article{chang-kopelowitz-pettie-2019-an-exponential-separation,
  author = {Yi-Jun Chang and Tsvi Kopelowitz and Seth Pettie},
  title = {An Exponential Separation between Randomized and Deterministic Complexity in the {LOCAL} Model},
  journal = {SIAM Journal on Computing},
  volume = {48},
  number = {1},
  pages = {122--143},
  year = {2019},
  doi = {10.1137/17M1117537}
}

@article{chang-pettie-2019-a-time-hierarchy-theorem-for-the,
  author = {Yi-Jun Chang and Seth Pettie},
  title = {A Time Hierarchy Theorem for the {LOCAL} Model},
  journal = {SIAM Journal on Computing},
  volume = {48},
  number = {1},
  pages = {33--69},
  year = {2019},
  doi = {10.1137/17M1157957}
}

@inproceedings{coiteux-roy-d-amore-etal-2024-no-distributed-quantum,
  author = {Xavier Coiteux-Roy and Francesco D'Amore and Rishikesh Gajjala and Fabian Kuhn and Fran{\c{c}}ois {Le Gall} and Henrik Lievonen and Augusto Modanese and Marc-Olivier Renou and Gustav Schmid and Jukka Suomela},
  editor = {Bojan Mohar and Igor Shinkar and Ryan O'Donnell},
  title = {No Distributed Quantum Advantage for Approximate Graph Coloring},
  booktitle = {Proceedings of the 56th Annual {ACM} Symposium on Theory of Computing, {STOC} 2024, Vancouver, BC, Canada, June 24-28, 2024},
  pages = {1901--1910},
  publisher = {{ACM}},
  year = {2024},
  doi = {10.1145/3618260.3649679}
}

@article{d-amore-2025-on-the-limits-of-distributed-quantum,
  _nodoi = {true},
  author = {Francesco d'Amore},
  journal = {Bulletin of the EATCS},
  pages = {51--91},
  title = {On the limits of distributed quantum computing},
  url = {http://bulletin.eatcs.org/index.php/beatcs/article/view/829},
  volume = {145},
  year = {2025}
}

@inproceedings{dahal-d-amore-etal-2023-brief-announcement-distributed,
  author = {Sameep Dahal and Francesco D'Amore and Henrik Lievonen and Timoth{\'e} Picavet and Jukka Suomela},
  editor = {Rotem Oshman},
  title = {Brief Announcement: Distributed Derandomization Revisited},
  booktitle = {37th International Symposium on Distributed Computing, {DISC} 2023, October 10-12, 2023, L'Aquila, Italy},
  series = {LIPIcs},
  volume = {281},
  pages = {40:1--40:5},
  publisher = {Schloss Dagstuhl - Leibniz-Zentrum f{\"{u}}r Informatik},
  year = {2023},
  doi = {10.4230/LIPICS.DISC.2023.40}
}

@inproceedings{dhar-kujawa-etal-2024-local-problems-in-trees-across-a,
  author = {Anubhav Dhar and Eli Kujawa and Henrik Lievonen and Augusto Modanese and Mikail Muftuoglu and Jan Studen{\'y} and Jukka Suomela},
  editor = {Silvia Bonomi and Letterio Galletta and Etienne Rivi{\`e}re and Valerio Schiavoni},
  title = {Local Problems in Trees Across a Wide Range of Distributed Models},
  booktitle = {28th International Conference on Principles of Distributed Systems, {OPODIS} 2024, December 11-13, 2024, Lucca, Italy},
  series = {LIPIcs},
  volume = {324},
  pages = {27:1--27:17},
  publisher = {Schloss Dagstuhl - Leibniz-Zentrum f{\"{u}}r Informatik},
  year = {2024},
  doi = {10.4230/LIPICS.OPODIS.2024.27}
}

@inproceedings{fraigniaud-heinrich-kosowski-2016-local-conflict,
  author = {Pierre Fraigniaud and Marc Heinrich and Adrian Kosowski},
  editor = {Irit Dinur},
  title = {Local Conflict Coloring},
  booktitle = {{IEEE} 57th Annual Symposium on Foundations of Computer Science, {FOCS} 2016, 9-11 October 2016, Hyatt Regency, New Brunswick, New Jersey, {USA}},
  pages = {625--634},
  publisher = {{IEEE} Computer Society},
  year = {2016},
  doi = {10.1109/FOCS.2016.73}
}

@inproceedings{gavoille-kosowski-markiewicz-2009-what-can-be-observed,
  author = {Cyril Gavoille and Adrian Kosowski and Marcin Markiewicz},
  editor = {Idit Keidar},
  title = {What Can Be Observed Locally?},
  booktitle = {Distributed Computing, 23rd International Symposium, {DISC} 2009, Elche, Spain, September 23-25, 2009. Proceedings},
  series = {Lecture Notes in Computer Science},
  volume = {5805},
  pages = {243--257},
  publisher = {Springer},
  year = {2009},
  doi = {10.1007/978-3-642-04355-0_26}
}

@inproceedings{izumi-le-gall-2019-quantum-distributed-algorithm-for,
  author = {Taisuke Izumi and Fran{\c{c}}ois {Le Gall}},
  editor = {Peter Robinson and Faith Ellen},
  title = {Quantum Distributed Algorithm for the All-Pairs Shortest Path Problem in the {CONGEST-CLIQUE} Model},
  booktitle = {Proceedings of the 2019 {ACM} Symposium on Principles of Distributed Computing, {PODC} 2019, Toronto, ON, Canada, July 29 - August 2, 2019},
  pages = {84--93},
  publisher = {{ACM}},
  year = {2019},
  doi = {10.1145/3293611.3331628}
}

@inproceedings{le-gall-2022-quantum-distributed-computing,
  _nodoi = {true},
  author = {Fran{\c{c}}ois {Le Gall}},
  booktitle = {11th Workshop on Advances in Distributed Graph Algorithms (ADGA 2022)},
  title = {Quantum Distributed Computing},
  url = {https://adga-workshop.org/2022/legall.pdf},
  year = {2022}
}

@inproceedings{le-gall-magniez-2018-sublinear-time-quantum-computation,
  author = {Fran{\c{c}}ois {Le Gall} and Fr{\'e}d{\'e}ric Magniez},
  title = {Sublinear-Time Quantum Computation of the Diameter in {CONGEST} Networks},
  booktitle = {Proceedings of the 2018 ACM Symposium on Principles of Distributed Computing},
  publisher = {ACM},
  pages = {337--346},
  year = {2018},
  doi = {10.1145/3212734.3212744}
}

@inproceedings{le-gall-nishimura-rosmanis-2019-quantum-advantage-for,
  author = {Fran{\c{c}}ois {Le Gall} and Harumichi Nishimura and Ansis Rosmanis},
  editor = {Rolf Niedermeier and Christophe Paul},
  title = {Quantum Advantage for the {LOCAL} Model in Distributed Computing},
  booktitle = {36th International Symposium on Theoretical Aspects of Computer Science, {STACS} 2019, March 13-16, 2019, Berlin, Germany},
  series = {LIPIcs},
  volume = {126},
  pages = {49:1--49:14},
  publisher = {Schloss Dagstuhl - Leibniz-Zentrum f{\"{u}}r Informatik},
  year = {2019},
  doi = {10.4230/LIPICS.STACS.2019.49}
}

@article{linial-1992-locality-in-distributed-graph-algorithms,
  author = {Nathan Linial},
  title = {Locality in Distributed Graph Algorithms},
  journal = {SIAM Journal on Computing},
  volume = {21},
  number = {1},
  pages = {193--201},
  year = {1992},
  doi = {10.1137/0221015}
}

@inproceedings{maus-tonoyan-2020-local-conflict-coloring-revisited,
  author = {Yannic Maus and Tigran Tonoyan},
  editor = {Hagit Attiya},
  title = {Local Conflict Coloring Revisited: Linial for Lists},
  booktitle = {34th International Symposium on Distributed Computing, {DISC} 2020, October 12-16, 2020, Virtual Conference},
  series = {LIPIcs},
  volume = {179},
  pages = {16:1--16:18},
  publisher = {Schloss Dagstuhl - Leibniz-Zentrum f{\"{u}}r Informatik},
  year = {2020},
  doi = {10.4230/LIPICS.DISC.2020.16}
}

@inproceedings{miller-reif-1985-parallel-tree-contraction-and-its,
  author = {Gary L. Miller and John H. Reif},
  title = {Parallel Tree Contraction and Its Application},
  booktitle = {26th Annual Symposium on Foundations of Computer Science, Portland, Oregon, USA, 21-23 October 1985},
  pages = {478--489},
  publisher = {{IEEE} Computer Society},
  year = {1985},
  doi = {10.1109/SFCS.1985.43}
}

@article{naor-stockmeyer-1995-what-can-be-computed-locally,
  author = {Moni Naor and Larry J. Stockmeyer},
  title = {What Can be Computed Locally?},
  journal = {SIAM Journal on Computing},
  volume = {24},
  number = {6},
  pages = {1259--1277},
  year = {1995},
  doi = {10.1137/S0097539793254571}
}

@inproceedings{suomela-2020-landscape-of-locality-invited-talk,
  author = {Jukka Suomela},
  editor = {Susanne Albers},
  title = {Landscape of Locality (Invited Talk)},
  booktitle = {17th Scandinavian Symposium and Workshops on Algorithm Theory, {SWAT} 2020, June 22-24, 2020, T{\'{o}}rshavn, Faroe Islands},
  series = {LIPIcs},
  volume = {162},
  pages = {2:1--2:1},
  publisher = {Schloss Dagstuhl - Leibniz-Zentrum f{\"{u}}r Informatik},
  year = {2020},
  doi = {10.4230/LIPICS.SWAT.2020.2}
}

\clearpage
\appendix

\section{Preliminaries}
\label{sec:preliminaries}

In this section, we introduce all the basic definitions and notations that we will use throughout the paper, and we also give a brief overview of the models of computation and probabilistic models that we will be working with.

\subsubsection*{Graphs} 
We work with simple, undirected graphs \(G = (V, E)\), where \(V\) is the set of nodes and \(E\) is the set of edges.
When the sets of nodes and edges are not specified, we can refer to them as \(V(G)\) and \(E(G)\), respectively.
Given any two nodes \(u, v \in V\), we denote by \(\dist(u,v)\) the distance between \(u\) and \(v\) in \(G\), that is, the length of the shortest path between \(u\) and \(v\) in \(G\).
The distance between a node \(v\) and a set of nodes \(S \subseteq V\) is defined as \(\dist(v,S) = \min_{u \in S} \dist(v,u)\).
Similarly, the distance between two sets of nodes \(S, T \subseteq V\) is defined as \(\dist(S,T) = \min_{u \in S, v \in T} \dist(u,v)\).
The diameter of \(G\) is defined as \(\diam(G) = \max_{u,v \in V} \dist(u,v)\).
The \emph{eccentricity} of a node \(v\) is defined as the maximum length of a shortest path from \(v\) to any other node in \(G\), that is, \(\ecc(v) = \max_{u \in V} \dist(u,v)\).
For any node \(v \in V\), we denote by \(\neighborhood(v) = \{u \in V \, : \, \dist(u,v) = 1\}\) the \emph{open neighborhood} of \(v\) in \(G\). 
All the elements of \(\neighborhood(v)\) are called the \emph{neighbors} of \(v\).
We also define the \emph{closed neighborhood} of \(v\) as \(\neighborhood[v] = \neighborhood(v) \cup \{v\}\).
For any integer \(r \geq 0\), we denote by \(\neighborhood_r[v] = \{u \in V \, : \, \dist(u,v) \leq r\}\) the set of nodes that are at distance at most \(r\) from \(v\) in \(G\), and by \(\neighborhood_r(v) = \{u \in V \, : \, \dist(u,v) = r\}\) the set of nodes that are at distance exactly \(r\) from \(v\) in \(G\).
We extend the definition of neighborhood to sets of nodes as well, by defining \(\neighborhood(S) = \{u \in V \, : \, \dist(u,S) = 1\}\), \(\neighborhood[S] = \{u \in V \, : \, \dist(u,S) \le 1\}\), \(\neighborhood_r(S) = \{u \in V \, : \, \dist(u,S) = r\}\), and \(\neighborhood_r[v] = \{u \in V \, : \, \dist(u,S) \le r\}\) for any \(S \subseteq V\).
Given a set of nodes \(S \subseteq V\), we denote by \(G[S]\) the subgraph of \(G\) induced by the nodes of \(S\).
Given two graphs \(G = (V_G, E_G)\) and \(H = (V_H, E_H)\), we say that \(G\) and \(H\) are isomorphic if there exists a bijection \(\varphi: V_G \to V_H\) such that \(\{u,v\} \in E_G\) if and only if \(\{\varphi(u),\varphi(v)\} \in E_H\).
The union of two graphs \(G = (V_G, E_G)\) and \(H = (V_H, E_H)\) is defined as the graph \(G \cup H = (V_G \cup V_H, E_G \cup E_H)\).
The difference of two graphs \(G = (V_G, E_G)\) and \(H = (V_H, E_H)\) is defined as the graph \(G \setminus H = (V_G, E_G \setminus E_H)\).

\subsubsection*{Half-edges}
Consider any graph \(G = (V,E)\).
We define the set of half-edges of \(G\) as the set \(\halfEdges(G) = \{(v,e) \, : \, v \in V, e \in E, v \in e\}\).
Any element of \(\halfEdges(G)\) is called a \emph{half-edge} of \(G\), and it can be thought of as the part of the edge that is incident to only one of its endpoints.
Hence, any edge \(e = \{u,v\} \in E\) corresponds to two half-edges, namely \((u,e)\) and \((v,e)\).

\subsubsection*{Edge-labeled graph}
Let \(\Sigma\) be a set of labels.
A \(\Sigma\)-edge-labeled graph is a tuple \((G = (V,E), \lambda)\) where \(G\) is a graph, and \(\lambda: \halfEdges(G) \to \Sigma\) is a function that assigns a label from \(\Sigma\) to each half-edge in \(\halfEdges(G)\).
The function \(\lambda\) is called the \emph{labeling} of \(G\).

\subsubsection*{Centered graphs} 
Let \(r \ge 0\) be an integer.
A centered graph of radius \(r\) is a tuple \((G, v)\) where \(G\) is a graph and \(v \in V(G)\) is a node such that the eccentricity of \(v\) in \(G\) is at exactly \(r\), that is, \(\ecc(v) = r\), with the following properties: 
First, we have \(\neighborhood_r[v] = V(G)\).
Second, if \(u_1,u_2 \in V(G)\) are such that \(\dist(v,u_1) = \dist(v,u_2) = r\), then \(u_1\) and \(u_2\) are not adjacent in \(G\). 
Note that this matches exactly what an $r$-round \local algorithm can gather about the graph at every node.

\subsubsection*{Set of constraints}
Let \(\Sigma\) be a set of labels, and let \(r, \Delta \ge 0\) be integers.
An \((r,\Delta)\)-set of constraints over \(\Sigma\) is a family \(\calC = \{(G_i,\lambda_i,v_i)\}_{i \in I}\) where, for each \((G_i, \lambda_i, v_i) \in \calC\), we have that \((G_i,v_i)\) is a centered graph of radius at most \(r\) and maximum degree at most \(\Delta\), \(\lambda_i \colon \calH(G_i) \to \Sigma\) is a labeling function, and \((G_i,\lambda_i)\) is a \(\Sigma\)-edge-labeled graph.

A \(\Sigma\)-edge-labeled graph \((G,\lambda)\) is said to \emph{satisfy} \(\calC\) if for any node \(v \in V(G)\), there exists some \((G_i,\lambda_i,v_i) \in \calC\) such that the \(G[\neighborhood_r[v]]\setminus G[\neighborhood_r(v)]\) in \(G\) is isomorphic to \(G_i\), and the isomorphism maps \(v\) to \(v_i\) and preserves the half-edge labels.

\subsubsection*{Locally checkable labeling problem}
Let \(\Sigma_{\textrm{in}}\), \(\Sigma_{\textrm{out}}\) be two finite sets of labels, and \(r,\Delta \ge 0\) be integers.
A \emph{locally checkable labeling} (LCL) problem \(\Pi\) is a tuple \((\Sigma_{\textrm{in}},\Sigma_{\textrm{out}}, r, \Delta, \calC)\), where \(\calC\) is a finite \((r,\Delta)\)-set of constraints over \(\Sigma_{\textrm{in}} \times \Sigma_{\textrm{out}}\).

\emph{Solving} \(\Pi\) means the following: 
given as input a \(\Sigma_{\textrm{in}}\)-edge-labeled graph \((G,\lambda_{\textrm{in}})\), where \(\lambda_{\textrm{in}}\colon \calH(G) \to \Sigma_{\textrm{in}}\) is the \emph{input} labeling of the graph, the goal is to produce an \emph{output} labeling \(\lambda_{\textrm{out}} \colon \calH(G) \to \Sigma_{\textrm{out}}\) such that, having defined \(\lambda \colon \calH(G) \to \Sigma_{\textrm{in}} \times \Sigma_{\textrm{out}}\) with \(\lambda(v,e) = (\lambda_{\textrm{in}}(v,e), \lambda_{\textrm{out}}(v,e))\),
we have that \((G, \lambda)\) satisfies \(\calC\).

\subsubsection*{Node-edge checkable problems}
We define an alternative notion of locally checkable labeling problems which at first seems more restrictive than the one defined above, but is in fact equivalent to it \cite{balliu-censor-hillel-etal-2021-locally-checkable}.
Let \(\Sigma_{\textrm{in}}\), \(\Sigma_{\textrm{out}}\) be two finite sets of labels, and \(\Delta \ge 0\) be an integer.
We specify the LCL problem through two sets of constraints \(\calC_V\) and \(\calC_E\).
The first set \(\calC_V\) is a set of elements, which we call \emph{node constraints}, where each element is a multiset of at most of at most \(\Delta\) elements, and each element belongs to \(\Sigma_{\textrm{in}} \times \Sigma_{\textrm{out}}\).
The second set \(\calC_E\) is a set of elements, which we call \emph{edge constraint}, where each element is a multiset of two elements, and each element of the multiset belongs to \(\Sigma_{\textrm{in}} \times \Sigma_{\textrm{out}}\).
Hence, the LCL problem can be defined via the tuple \((\Sigma_{\textrm{in}},\Sigma_{\textrm{out}}, \Delta, \calC_V, \calC_E)\).

\emph{Solving} \(\Pi\) means the following:
given as input a \(\Sigma_{\textrm{in}}\)-edge-labeled graph \((G,\lambda_{\textrm{in}})\), where \(\lambda_{\textrm{in}}\colon \calH(G) \to \Sigma_{\textrm{in}}\) is the \emph{input} labeling of the graph, the goal is to produce an \emph{output} labeling \(\lambda_{\textrm{out}} \colon \calH(G) \to \Sigma_{\textrm{out}}\) such that the following holds:
\begin{itemize}
    \item For each node \(v \in V(G)\), the multiset of labels \(\{\lambda(v,e) \, : \, e \in E(G), v \in e\}\) belongs to \(\calC_V\).
    \item For each edge \(e = \{u,v\} \in E(G)\), the multiset of labels \(\{\lambda(u,e), \lambda(v,e)\}\) belongs to \(\calC_E\).
\end{itemize}
It has been proven that the two definitions of the LCL problems are equivalent in trees for the LOCAL model, in the sense that, for any LCL problem defined according to the first definition, there exists an equivalent LCL problem defined according to the second definition, and vice versa, and there are constant-time LOCAL algorithms that can transform any solution of the first problem into a solution of the second problem, and vice versa \cite{balliu-censor-hillel-etal-2021-locally-checkable}.
For simplicity, in the rest of the paper we will just work with node-edge checkable problems.

\subsection{Models of computation}

\subsubsection*{\texorpdfstring{The \local model}{The LOCAL model}}
We work in the classical \local model of distributed computation 
\cite{linial-1992-locality-in-distributed-graph-algorithms}. 
In this model, we are given \(n\) computing units (called \emph{nodes}) that are connected by a communication network.
The communication network is modeled as an undirected graph \(G = (V, E)\) where the nodes of the graph correspond to the computing units, and the edges correspond to communication links between the nodes.
The nodes run the same algorithm and operate in synchronous rounds, and in each round, the nodes perform the following steps: 
First, they perform some local computation based on their current state and the information they have received from their neighbors in the previous rounds.
Then, they send messages to and receive messages from their neighbors.
Note that in the \local model there is no limit on the size of the messages that can be sent in each round, and there is no limit on the amount of local computation that can be performed by the nodes in each round.
At the beginning of computation, all nodes are identical, except for a unique identifier of \(O(\log n)\) bits, which is used to break symmetry in the computation, their degree, and the port numbers of their incident edges.
The nodes may also have access to some additional (distinguished) input which is part of the problem specification.
The computation ends when all nodes have decided on their output, and the running time of the algorithm is the number of communication rounds until all nodes have decided on their output.
The complexity of a problem in the \local model is measured by the minimum number of rounds of any algorithm that solves the problem in the \local model.
Nodes may be given access to independent sources of infinite random bits strings, in which case we talk about the randomized \local model as opposed to the deterministic \local model.
In the randomized \local model, we require that a problem is solved with high probability, that is, for any given \(c > 0\), the algorithm must solve the problem with probability at least \(1 - 1/n^c\).

\subsubsection*{Input graph} 
As already mentioned, the input graph is the network communication graph.
Nodes have unique identifiers of size \(O(\log n)\) bits, where \(n\) is the number of nodes in the network, and can distinguish the communication links towards their neighbors through port numbers from \(1\) up to the degree of the node. 
Both identifiers and port numbers are adversarially chosen and part of the input graph, unrelated to the specific problem we are solving.
The input labels of the problem of interest, the identifiers and port numbers, and the inital states of the nodes, possibly together with the random bits, are encoded in some local variable \(\inpt(v)\), where \(v \in V(G)\).

\subsubsection*{Local view}
Let \(r \ge 0\) be any integer.
Given a graph \((G, \inpt)\) as input to some problem in the \local model, for any \(v \in V(G)\), we define the \emph{radius-\(r\) view of \(v\)}, denoted by \(\view_r(v,G,\inpt)\), as the graph \(G_r(v) = G[\neighborhood_r[v]]\setminus G[\neighborhood_r(v)]\), where each node \(u \in V(G_r(v))\) is equipped with its own local variable \(\inpt(u)\).
The radius-\(r\) view of a subset of nodes \(A \subseteq V(G)\) is analogously defined, and is denoted by \(\view_r(A,G,\inpt)\).
Given two input graphs \((G,\inpt_G)\) and \((H, \inpt_H)\), and two subsets of nodes \(A \subseteq V(G)\) and \(B \subseteq B(H)\), we say that  \(\view_r(A,G,\inpt_G)\) is isomorphic to \(\view_r(B,H,\inpt_H)\) 
if there exists an isomorphism \(\varphi\) between \(G_r(A)\) and \(H_r(B)\) such that the following holds:
\begin{enumerate}
    \item \(\varphi \restriction_A\) is an isomorphism between \(G[A]\) and \(H[B]\).
    \item \(\inpt_H(\varphi(v)) = \inpt_G(v)\) for each \(v \in V(G_r(A))\).
\end{enumerate}

\subsubsection*{\texorpdfstring{The quantum-\local model}{The quantum-LOCAL model}}

The quantum-\local model is a quantum extension of the \local model, where the nodes are quantum computing units that exchange quantum messages.
More specifically, the nodes can manipulate an arbitrary number of qubits by applying unitaries, perform any kind of measurements on their local states, and exchange messages of arbitrary size with their neighbors.
The output of the algorithm is determined by the measurement outcomes of the nodes at the end of the computation, and the complexity of a problem in the quantum-\local model is measured by the minimum number of rounds of any quantum algorithm that solves the problem in the quantum-\local model.
Similarly to the randomized \local model, we require that a problem is solved with high probability in the quantum-\local model as well.
For a formal definition of the model, see \cite{gavoille-kosowski-markiewicz-2009-what-can-be-observed}.

\subsection{Probabilistic models}

Note that randomized and quantum-\local algorithms produce probability distributions over output labelings of the input graph.
We formalize this concept following the definitions of Akbari et al.\ \cite{akbari-coiteux-roy-etal-2025-online-locality-meets}.

\subsubsection*{Outcomes}
Let \(\Sigma\) be a set of labels.
An outcome is a function \(\outcome \colon (G,\inpt) \mapsto \{(\lambda_i, p_i)\}_{i \in I} \) mapping each input graph \((G,\inpt)\) to a probability distribution over labeling functions, that is, a family \(\{(\lambda_i, p_i)\}_{i \in I}\) where \(I \subseteq \reals\) is any Lebesgue-measurable set of indices, \(\lambda_i \colon \calH(G) \to \Sigma\) is a labeling function, and \(p_i\) is the probability that \(G\) is actually labeled by \(\lambda_i\).
Note that it must hold that \(\sum_{i \in I} p_i = 1\), and that each \(p_i\) is non-negative.
If \(I\) is uncountable, \(p\colon I \to [0,1]\) is a measurable density function and the sum is meant to be an integral, that is, \(\int_{I} p(x) \dd{x} = 1 \).

We remark that an outcome must be defined \emph{for every input graph}, just like a \local algorithm can be executed on every input graph.
Indeed, if at any node of a graph and at any round, the computation specified by an algorithm is undefined, we can assume that the node can always output some \emph{garbage label} \(\bot\).

Given a subgraph \(G' \subseteq G\), the \emph{restriction} 
\(\outcome(G,\inpt) \restriction_{(G', \inpt \restriction_{V(G')})}\) 
of the 
probability distribution \(\outcome(G,\inpt) = \{(\lambda_i, p_i)\}_{i \in I}\) to \((G', \inpt \restriction_{V(G')})\) is the probability distribution \(\{(\eta_j, q_j)\}_{j \in J}\) where \(J \subseteq I\), \(\eta_j \colon \calH(G') \to \Sigma\) is such that there exists \(i \in I\) with \(\lambda_i \restriction_{\calH(G')} = \eta_j\), and 
\[
    q_j = \sum_{i \in I \, : \, \lambda_i \restriction_{\calH(G')} = \eta_j} p_i.
\]
The sum is replaced by an integral when suitable.

\subsubsection*{Non-signaling outcomes}
Let \(\outcome\) be any outcome.
We say that \(\outcome\) is non-signaling beyond distance \(T = T(n)\) if, for any two input graphs \((G, \inpt_G)\) and \((H,\inpt_H)\) of \(n\) nodes such that there exists two subsets of nodes \(A \subseteq V(G)\) and \(B \subseteq V(H)\) with \(\view_T(A,G,\inpt_G)\) that is isomorphic to \(\view_T(A,H,\inpt_H)\), then the probability distributions \(\outcome(G,\inpt_G) \restriction_{(G[A], \inpt_G \restriction_A)}\) and \(\outcome(H,\inpt_H) \restriction_{(H[B], \inpt_H \restriction_B)}\) are the same.

\subsubsection*{\texorpdfstring{\(\mathbf{T}\)-dependent distributions}{T-dependent distributions}}
Given an input graph \((G, \inpt_G)\) and a probability distribution \(\{(\lambda_i, p_i)\}_{i \in I}\) over labelings of \(G\), we say that \(\{(\lambda_i, p_i)\}_{i \in I}\) is \(T\)-dependent if, for any two subsets of nodes \(A,B \subseteq V(G)\) such that \(\dist(A,B) > T\), the probability distributions obtained by the restrictions of \(\{(\lambda_i, p_i)\}_{i \in I}\) to \((G[A], \inpt \restriction_A)\) and to \((G[B], \inpt \restriction_B)\) are independent.

\subsubsection*{\texorpdfstring{The bounded-dependence model}{The bounded-dependence model}}
In the bounded-dependence model, the input graph is as in the classical \local model, and we produce outcomes.
We say that an outcome \(\outcome\) in the bounded-dependence model has locality \(T = T(n)\) if the following properties are satisfied:
\begin{enumerate}
    \item \(\outcome\) is non signaling beyond distance \(T\).
    \item Given an input graph \((G,\inpt)\), the probability distribution \(\outcome(G,\inpt)\) is \(2T\)-dependent.
\end{enumerate}
If \(T(n) = O(1)\) independently of \(n\), we say that the probability distributions generated by \(\outcome\) are \emph{finitely-dependent}.
Similarly to the classical \local model, an outcome is required to solve the problem with high probability.

\section{Hierarchical decomposition of a tree}\label{sec:tree-decomposition}

In this section, we show how to perform a hierarchical decomposition of a tree into smaller pieces.
At the end of this decomposition, the nodes of the tree are decomposed into a sequence of pairwise disjoint sets \(\VR_{1}, \VC_{1}, \VR_{2}, \VC_{2}, \ldots\), where the subscript and superscript together indicate the \emph{level} of the decomposition, and the superscript indicates the \emph{type} of the level, which is either \emph{rake} (denoted by \(\mathsf{R}\)) or \emph{compress} (denoted by \(\mathsf{C}\)).
We denote the forest induced by nodes $\VR_i \cup \VC_i \cup \VR_{i+1} \cup \VC_{i+1} \cup \cdots$ by $\GR_i$, and the forest induced by nodes $\VC_i \cup \VR_{i+1} \cup \VC_{i+1} \cup \cdots$ by $\GC_i$.
Let $L$ be the smallest integer \(i\) such that $\GC_i$ is the empty graph.

Given two parameters $\ell \ge 1$ and $\gamma \ge 1$, the decomposition satisfies the following two key properties:
\begin{enumerate}
	\item Each connected component of the subgraph induced by $\VR_i$ must be a rooted tree of height at most $\gamma-1$ such that the root has at most one neighbor in $\GC_i$, while other nodes have no neighbors in $\GC_i$.
	\item Each connected component of the subgraph induced by $\VC_i$ must be a path on a number of nodes in the range $[\ell, 2\ell]$ such that the endpoints have at most one neighbor in $\GR_{i+1}$, while intermediate nodes have no neighbors in \(\GR_{i+1}\). If the path has one node, then this single node can have at most two neighbors in \(\GR_{i+1}\).
\end{enumerate}
A decomposition satisfying these properties is called a \emph{$(\gamma, \ell)$-decomposition}.
We order the layers as
\(\VR_1,\VC_1,\VR_2,\VC_2,\ldots\).  For later use, for each connected
component \(C\) of a layer, we define the \emph{upward-component of \(C\)},
denoted \(C^+\), as the union of all components reachable from \(C\) by a path
that visits components in strictly increasing layer order.  Equivalently,
starting from \(C\), we repeatedly add neighboring components that belong to
higher layers.

\subsection{A rake-and-compress algorithm for the decomposition}

We now describe the decomposition procedure.  
Algorithm~\ref{alg:nice-rake-compress} is described in a centralized manner.
The purpose of this centralized description is to specify the shape of the
levels.  
Afterward, we explain why the same phase structure can be implemented in the deterministic \local model.

The algorithm proceeds in phases.  
At the beginning of a phase \(i\), \(F\) is the forest induced by the nodes that have not yet received a rake or compress level.
The parameter \(\gamma\) is the length scale used in the rake step, while
\(\ell\) is the length scale used in the compress step.

The phase starts with a \emph{rake step}.  
The rake step consists of \(\gamma\) rake operations.  In each operation, all
current isolated nodes and leaves of the remaining forest are removed, except
that if a connected component is a single edge, then only the endpoint with
larger identifier is removed.  
The union of the nodes removed by these \(\gamma\) operations is the next rake
level \(\VR_i\).

After the rake step, if any nodes remain, the phase performs a
\emph{compress step}.  
The compress
step scans each maximal degree-\(2\) path, meaning a maximal path whose nodes
all have degree \(2\) in the current forest, and removes long intervals between
kept boundary nodes.  Each removed interval contains between \(\ell\) and
\(2\ell\) nodes.  
The boundary nodes remain in the forest, so only the endpoints of the removed
pieces can have neighbors in the next forest.

\begin{algorithm}[t]
\caption{\(\textsc{Tree-decomposition}\)}
\label{alg:nice-rake-compress}
\begin{algorithmic}[1]
\Require A forest \(G=(V,E)\), integers \(\gamma,\ell \ge 1\)
\Ensure A \((\gamma,\ell)\)-decomposition \((\VR_1,\VC_1,\ldots,\VR_L,\VC_L)\) of \(G\), with \(L = O(\lceil \ell/\gamma\rceil\log n)\)

\State \(F \gets G\)
\State \(i \gets 1\)

\While{\(V(F)\neq \emptyset\)}

    \Comment{Rake step}
    \State \(\VR_i \gets \emptyset\)
    \State \(F_{\mathsf{rake}} \gets F\)
    \For{\(t=1,2,\ldots,\gamma\)}
        \State \(R_t \gets \textsc{RakeStep}(F_{\mathsf{rake}})\)
        \State \(\VR_i \gets \VR_i \cup R_t\)
        \State \(F_{\mathsf{rake}} \gets F_{\mathsf{rake}}[V(F_{\mathsf{rake}})\setminus R_t]\)
    \EndFor
    \State \(F \gets F_{\mathsf{rake}}\)

    \If{\(V(F)=\emptyset\)}
        \State \(\VC_i \gets \emptyset\)
        \State \textbf{break}
    \EndIf

    \Comment{Compress step}
    \State \(\VC_i \gets \textsc{CompressPaths}(F,\ell)\)
    \State \(F \gets F[V(F)\setminus \VC_i]\)
    \State \(i \gets i+1\)

\EndWhile

\State \(L \gets i\)
\State \Return \((\VR_1,\VC_1,\ldots,\VR_L,\VC_L)\)

\end{algorithmic}
\end{algorithm}

\begin{algorithm}[t]
\caption{\(\textsc{RakeStep}(F)\)}
\label{alg:rake-step}
\begin{algorithmic}[1]
\Require A forest \(F\)
\Ensure A set \(R\subseteq V(F)\) of nodes to be removed by one rake operation

\State \(R \gets \{v\in V(F): \deg_F(v)=0\}\)
\State \(R \gets R \cup \{v\in V(F): v\text{ has unique neighbor }u,\text{ and either }\deg_F(u)\ge 2\text{ or }\mathrm{id}(v)>\mathrm{id}(u)\}\)
\State \Return \(R\)

\end{algorithmic}
\end{algorithm}

\begin{algorithm}[t]
\caption{\(\textsc{CompressPaths}(F,\ell)\)}
\label{alg:compress-paths}
\begin{algorithmic}[1]
\Require A forest \(F\), an integer \(\ell\ge 1\)
\Ensure A set \(X\subseteq V(F)\) of nodes to be removed as compress components

\State \(X \gets \emptyset\)

\ForAll{maximal paths \(P=v_1\ldots v_k\) whose nodes all have degree \(2\) in \(F\)}
    \If{\(k-2 \ge \ell\)}
        \State \(s \gets 1\)

        \While{\(k-s-1 \ge \ell\)}
            \If{\(k-s-1 \le 2\ell\)}
                \State \(q \gets k\)
            \Else
                \State \(q \gets s+2\ell+1\)
            \EndIf

            \State \(X \gets X \cup \{v_j : s<j<q\}\)
            \State \(s \gets q\)
        \EndWhile
    \EndIf
\EndFor

\State \Return \(X\)

\end{algorithmic}
\end{algorithm}

\subsection{Correctness of the decomposition algorithm}

\begin{lemma}\label{lem:tree-decomposition-partition}
    If Algorithm~\ref{alg:nice-rake-compress} terminates, it returns a
    partition of \(V(G)\).
\end{lemma}
\begin{proof}
    Each nonempty set \(\VR_i\) or \(\VC_i\) is a subset of the current forest
    \(F\), and after assigning it, the algorithm replaces \(F\) by the induced
    forest on the remaining nodes.  
    Hence, no node can be assigned to two different
    levels.  
    The algorithm terminates only after \(F\) is empty, because the last rake
    removal emptied the forest.  
    Therefore, every node is assigned
    to exactly one level.
\end{proof}

With the next lemma, we prove that the rake components satisfy the required properties.

\begin{lemma}\label{lem:tree-decomposition-rake-components}
    Let \(C\) be a connected component of \(G[\VR_i]\).  Then \(C\) can be
    rooted with height at most \(\gamma-1\) so that the root has at most one
    neighbor in \(\GC_i\), while all other nodes of \(C\) have no neighbors
    in \(\GC_i\).
\end{lemma}
\begin{proof}
    Let \(F_i\) be the forest at the beginning of phase \(i\).  The nodes
    of \(\GC_i\) are precisely the nodes of \(F_i-\VR_i\).

    For every node \(v\in C\), let
    \(t(v)\in\{1,\ldots,\gamma\}\) be the rake operation in which \(v\) is
    removed.  At the moment \(v\) is removed, it has degree at most \(1\) in the
    current forest.  Hence \(v\) has at most one neighbor whose removal time is
    larger than \(t(v)\), or which remains after the rake step.
    Moreover, two adjacent nodes cannot be removed in the same rake operation:
    the only case in which two adjacent nodes are both leaves is a single-edge
    component, and then the identifier tie-break removes only one endpoint.

    In a connected component \(C\) of \(G[\VR_i]\), root \(C\) at the unique
    node with no neighbor in \(C\) of larger removal time.  Uniqueness follows
    because \(G[C]\) is a tree: if two such nodes existed, the path between them
    inside \(C\) would contain a node with two neighbors that survive longer,
    contradicting that this node was a leaf when removed.  
    Along every path towards the root, removal times strictly increase, and therefore the height
    of \(C\) is at most \(\gamma-1\).

    The root is the only node of \(C\) that can have a neighbor in
    \(F_i-\VR_i\), and it has at most one such neighbor, while all other nodes of
    \(C\) have no neighbors in \(F_i-\VR_i\).
\end{proof}

The next lemma proves that the compress components satisfy the required properties.
\begin{lemma}\label{lem:tree-decomposition-compress-components}
    Let \(C\) be a connected component of \(G[\VC_i]\).  Then \(C\) is a path
    on a number of nodes in the range \([\ell,2\ell]\).  If \(C\) has at least two
    nodes, then the endpoints have at most one neighbor in \(\GR_{i+1}\),
    while intermediate nodes have no neighbors in \(\GR_{i+1}\).  If \(C\)
    consists of a single node, then this node has at most two neighbors in
    \(\GR_{i+1}\).
\end{lemma}
\begin{proof}
    If \(\VC_i=\emptyset\), there is nothing to prove.  Otherwise, let \(F_i\)
    be the forest just before \(\VC_i\) is removed.  
    A component
    \(C\) produced by the compress step is the interval
    \(v_{s+1}\ldots v_{q-1}\) inside some maximal degree-\(2\) path
    \(P=v_1\ldots v_k\) of \(F_i\).  All nodes of \(P\) have degree \(2\) in
    \(F_i\).  
    The nodes \(v_s\) and \(v_q\)
    are not added to \(\VC_i\), and hence they remain in \(F_i-\VC_i\).  
    They are
    adjacent to \(C\), and they are the only possible neighbors of \(C\) in
    \(F_i-\VC_i\).  
    The nodes of \(F_i-\VC_i\) are exactly the nodes that have not yet
    been assigned after the compress step, and these induce \(\GR_{i+1}\).
    If \(C\) has at least two nodes, then its endpoints are adjacent to at
    most one of \(v_s\) and \(v_q\), while its intermediate nodes are
    adjacent to neither.  If \(C\) consists of the single node
    \(v_{s+1}=v_{q-1}\), then this node is adjacent to both \(v_s\) and
    \(v_q\), and has no other neighbor in \(\GR_{i+1}\).  Hence it has at most
    two neighbors in \(\GR_{i+1}\).

    The compress step chooses \(q\) so that
    \(\ell\le q-s-1\le 2\ell\).  Thus the removed interval \(C\) contains
    \(q-s-1\) nodes, which is in the range \([\ell,2\ell]\).
\end{proof}

Wrapping up, we conclude with the following lemma.
\begin{lemma}
    The output of Algorithm~\ref{alg:nice-rake-compress} is a
    \((\gamma,\ell)\)-decomposition of \(G\).
\end{lemma}
\begin{proof}
    By \cref{lem:tree-decomposition-partition}, the output is a partition.
    \Cref{lem:tree-decomposition-rake-components} proves
    property~\ref{prop:tree-decomposition-rake} with parameter \(\gamma\), and
    \cref{lem:tree-decomposition-compress-components} proves
    property~\ref{prop:tree-decomposition-compress} with parameter \(\ell\).
    Thus, every condition in the definition of a
    \((\gamma,\ell)\)-decomposition is satisfied.
\end{proof}

\subsubsection*{Progress}

The next lemma shows that compress steps do not leave very long degree-\(2\) paths,
which is a key property for the progress of the algorithm.
\begin{lemma}
    \label{lem:compress-short-degree-two-paths}
    Let \(F'\) be the forest obtained after applying the compress step to a
    forest \(F\).  Then every maximal degree-\(2\) path of \(F'\) contains fewer
    than \(\ell+2\) nodes.
\end{lemma}
\begin{proof}
    Let \(P=v_1\ldots v_k\) be a maximal degree-\(2\) path of \(F\).  If
    \(k-2<\ell\), then \(P\) contains fewer than \(\ell+2\) nodes, and the
    compress step does not modify \(P\).  If \(k-2\ge\ell\), the compress step
    removes the internal nodes of intervals \(v_s\ldots v_q\) with
    \(q-s-1\ge\ell\).  The nodes \(v_s\) and \(v_q\) are kept, but after the
    internal nodes are removed they are no longer connected through that
    interval.  Thus, every remaining connected subpath of \(P\) contains fewer
    than \(\ell+2\) nodes that still have degree \(2\).  Since the first and
    last nodes of every compressed degree-\(2\) path are kept, a node outside
    \(P\) does not lose its incident edge to \(P\) during this compression.
    Hence, separate degree-\(2\) paths cannot merge through such a node.
\end{proof}

For a forest \(F\), let \(S(F)\) be the forest obtained by contracting every
maximal degree-\(2\) path of \(F\) into one edge.  We call \(S(F)\) the
\emph{skeleton} of \(F\).  By construction, \(S(F)\) has no degree-\(2\) nodes.
With the next lemma, we show that sequences of rake steps remove all skeleton nodes
of degree at most \(1\), which is a key property for the progress of the
algorithm.

\begin{lemma}
    \label{lem:rake-deletes-skeleton-leaves}
    Suppose that every maximal degree-\(2\) path of a forest \(F\) contains
    fewer than \(\ell+2\) nodes.  Let
    \(h=\lceil (\ell+3)/\gamma\rceil\).  After
    \(h\) consecutive rake steps, every node of \(S(F)\) of degree at most \(1\)
    has disappeared from the current skeleton.
\end{lemma}
\begin{proof}
    Let \(x\) be a node of \(S(F)\) of degree at most \(1\).  If \(x\) is
    isolated in \(S(F)\), then it is isolated in \(F\) and is removed by the
    first rake operation.  Otherwise, in \(F\), the skeleton edge incident to
    \(x\) corresponds to a path
    \[
        xu_1\ldots u_ry,
    \]
    where \(r\ge 0\), each node \(u_j\) has degree \(2\), and \(r<\ell+2\).
    If the component is just this path, choose either endpoint as \(x\).

    A rake step consists of \(\gamma\) leaf-removal operations.  As long as
    some node of \(xu_1\ldots u_r\) remains, the first remaining node on the path
    \(xu_1\ldots u_r\) is a leaf of the current forest.  Hence, one rake step
    removes the next \(\gamma\) nodes of this list, unless the whole list has
    already been removed.  The only exception is when the remaining component is
    a single edge and the first node of the list has the smaller identifier; in
    that case the other endpoint is removed first, and the first node is removed
    in the next leaf-removal operation.  After
    \(h=\lceil (\ell+3)/\gamma\rceil\) rake steps, all nodes
    \(xu_1\ldots u_r\) have been removed, because
    \(r+2\le \ell+3\le h\gamma\).  Thus, the old skeleton node \(x\), and the
    path connecting it to the rest of the skeleton, no longer appear in the
    current skeleton.
\end{proof}

If we remove all nodes of degree at most \(1\) from the skeleton of a forest,
the number of remaining nodes decreases by a constant factor, as shown by the
next lemma. 

\begin{lemma}
    \label{lem:leaf-deletion-shrinks-skeleton}
    Let \(S\) be a forest with no degree-\(2\) nodes.  If all nodes of degree at
    most \(1\) in \(S\) are deleted, the number of remaining nodes is at most
    half of \(\abs{V(S)}\).
\end{lemma}
\begin{proof}
    Let \(L\) be the set of nodes of degree at most \(1\) in \(S\), and let
    \(B\) be the set of nodes of degree at least \(3\).  Since \(S\) has no
    degree-\(2\) nodes, \(V(S)=L\cup B\).  In every forest,
    \(\abs{B}\le \abs{L}\).  Hence, after deleting \(L\), at most
    \(\abs{B}\le \abs{V(S)}/2\) nodes remain.
\end{proof}

The following result concludes the analysis of the progress of the algorithm.

\begin{lemma}\label{lem:progress-lemma}
    Algorithm~\ref{alg:nice-rake-compress} creates
    \(O(\lceil \ell/\gamma\rceil\log n)\) levels.
\end{lemma}
\begin{proof}
    Let \(h=\lceil (\ell+3)/\gamma\rceil\).  After every compress step, the current
    forest has no maximal degree-\(2\) path with at least \(\ell+2\) nodes by
    \cref{lem:compress-short-degree-two-paths}.  If the algorithm does not
    terminate during the first phase, its first compress step establishes this
    property.  Starting from such a forest,
    \cref{lem:rake-deletes-skeleton-leaves} says that \(h\) rake steps would
    delete all nodes of degree at most \(1\) in the skeleton present at the
    beginning of the block.  The compress steps performed inside the block only
    pass to subforests of the forest that would be obtained without compression:
    they delete nodes and incident edges, and never add nodes or edges.  Thus
    degrees can only decrease, and any node that would become a leaf after some
    number of leaf-removal operations without these compress steps becomes a
    leaf no later if the compress steps are also performed, unless it has
    already been deleted by a compress step.

    Thus, after the first phase, every block of \(h\) phases removes at least
    all nodes of degree at most \(1\) in the skeleton present at the beginning of
    the block.  By \cref{lem:leaf-deletion-shrinks-skeleton}, deleting these
    nodes leaves at most half of the skeleton nodes.  Hence, if \(s_j\) is the
    skeleton size at the beginning of the \(j\)-th such block, then
    \(s_{j+1}\le s_j/2\).  Since the initial skeleton has at most
    \(n\) nodes, there are at most
    \(1+h\lceil \log_2 n\rceil\) phases.  Each phase creates at most one rake
    level and one compress level, so the number of levels is at most
    \(2+2h\lceil \log_2 n\rceil =
    O(\lceil \ell/\gamma\rceil\log n)\).
\end{proof}

We now show that \cref{alg:nice-rake-compress} can be implemented in the deterministic \local model with the same phase structure, and with a cost of \(O(\gamma+\ell^2+\ell\log^* n)\) rounds per phase.
We first introduce the notion of ruling sets, which are a generalization of maximal independent sets.

\begin{definition}
    Let \(a,b\ge 1\) be integers.  A set \(S\subseteq V(G)\) is an
    \((a,b)\)-ruling set of \(G\) if every two distinct nodes of \(S\) are
    at distance at least \(a\), and every node of \(G\) is at distance at most
    \(b\) from some node of \(S\).
\end{definition}

Computing specific ruling sets can be done by computing an MIS in the power graph of the input graph.
Deterministic algorithms for MIS have been extensively studied in the literature \cite{fraigniaud-heinrich-kosowski-2016-local-conflict,barenboim-elkin-goldenberg-2022-locally-iterative,maus-tonoyan-2020-local-conflict-coloring-revisited}. 
We do not wish to optimize the running time of the ruling set, and the following lemma is enough for our purposes, as we will argue in the following sections.

\begin{lemma}
    \label{lem:path-power-mis}
    Let \(P_n\) be an \(n\)-node path, and let \(\alpha\ge 2\).  There is a
    deterministic \local algorithm that computes an
    \((\alpha,\alpha-1)\)-ruling set of \(P_n\) in
    \(O(\alpha^2+\alpha\log^* n)\) rounds.
\end{lemma}
\begin{proof}
    Let \(r=\alpha-1\), and let \(H=P_n^r\) be the graph on node set
    \(V(P_n)\) in which two distinct nodes are adjacent if their distance in
    \(P_n\) is at most \(r\).  The maximum degree of \(H\) is at most
    \(2r=O(\alpha)\).  We compute a deterministic MIS of \(H\) using a
    degree-dependent \(O(\Delta+\log^* n)\)-round deterministic MIS algorithm
    \cite{barenboim-elkin-goldenberg-2022-locally-iterative}.
    A communication round in \(H\) can be simulated in \(O(\alpha)\) rounds on
    the original path \(P_n\), because the neighbors of a node in \(H\) are
    within distance at most \(r=\alpha-1\) in \(P_n\).  Thus, the total cost on
    \(P_n\) is
    \(
        O(\alpha\cdot(\Delta(H)+\log^* n))
        =
        O(\alpha^2+\alpha\log^* n)
    \), where \(\Delta(H)\) is the maximum degree of \(H\).

    Let \(S\) be the MIS of \(H\).  Since \(S\) is independent in \(H\), no two
    nodes of \(S\) are at distance at most \(r\) in \(P_n\).  Hence, distinct
    nodes of \(S\) are at distance at least \(r+1=\alpha\) in \(P_n\).  Since
    \(S\) is maximal in \(H\), every node of \(P_n\) is either in \(S\) or has
    a neighbor in \(S\) in \(H\), and is therefore within distance
    \(r=\alpha-1\) from \(S\) in \(P_n\).  Thus, \(S\) is an
    \((\alpha,\alpha-1)\)-ruling set of \(P_n\).
\end{proof}

The following lemma establishes that one phase of our \cref{alg:nice-rake-compress} can be implemented in the deterministic \local model with a cost of \(O(\gamma+\ell^2+\ell\log^* n)\) rounds.
\begin{lemma}
    \label{lem:phase-local-implementation}
    One phase of the rake-and-compress decomposition can be implemented in
    \(O(\gamma+\ell^2+\ell\log^* n)\) rounds in the deterministic \local model.
\end{lemma}
\begin{proof}
    The rake level is obtained by simulating \(\gamma\) rounds of leaf removal
    in the current forest, using identifiers to break symmetry on isolated
    edges.  This takes \(O(\gamma)\) rounds in total.

    For the compress level, consider the subgraph induced by the degree-\(2\)
    nodes of the current forest.  Its connected components are precisely the
    maximal degree-\(2\) paths.  The LOCAL implementation of the compress step
    uses separators on these paths rather than the centralized scan of
    Algorithm~\ref{alg:compress-paths}.  Let \(P=v_1\ldots v_k\) be one such
    path.  If \(k-2<\ell\), no node of \(P\) is removed.  Otherwise, \(v_1\) and
    \(v_k\) are kept as separators.  The nodes within distance at most \(\ell\)
    from either endpoint are not allowed to become additional separators.  On
    the remaining middle subpath, if it is nonempty, we compute a deterministic
    \((\ell+1,\ell)\)-ruling set using \cref{lem:path-power-mis}, with
    \(\alpha=\ell+1\); this takes \(O(\ell^2+\ell\log^* n)\) rounds.

    Let \(S_P\) be the union of the two endpoints and the ruling-set nodes,
    ordered along \(P\).  Consecutive nodes of \(S_P\) are at distance at
    least \(\ell+1\): this follows from definition of ruling set, and from the
    fact that nodes within distance less than \(\ell+1\) from an
    endpoint were not allowed to become separators.  Consecutive nodes of
    \(S_P\) are also at distance at most \(2\ell+1\): between two ruling-set
    nodes this follows from the covering condition, and near an endpoint the
    first eligible node is at
    distance \(\ell+1\) from the endpoint and is within distance at most
    \(\ell\) from a ruling-set node.  
    If the middle subpath is empty, then
    \(k-2\le 2\ell\), so the two endpoints are at distance at most \(2\ell+1\).

    The distributed compress step removes exactly the internal nodes
    between each pair of consecutive separators in \(S_P\).  Thus, it may choose
    different separators from the centralized scan in
    Algorithm~\ref{alg:compress-paths}, but every removed component is still a
    path on between \(\ell\) and \(2\ell\) nodes, and its only possible
    neighbors in the next forest are the two separators.  These are precisely the
    properties used in the correctness and progress proofs.  Hence, one
    complete phase takes \(O(\gamma+\ell^2+\ell\log^* n)\) rounds.
\end{proof}

Now we are ready to conclude with the main result of this section, which states that a \((\gamma,\ell)\)-decomposition can be computed in the deterministic \local model with locality \(O((\gamma+\ell^2+\ell\log^* n)\lceil \ell/\gamma\rceil\log n)\).
\begin{theorem}
    \label{lem:tree-decomposition-local}
    Let \(\gamma,\ell\ge 1\) be integers.
    There exists a deterministic \local algorithm \(\calA\) that computes a
    \((\gamma,\ell)\)-decomposition of a forest of \(n\) nodes with
    \(L = O(\lceil \ell/\gamma\rceil\log n)\) levels in time
    \(O((\gamma+\ell^2+\ell\log^* n)\lceil \ell/\gamma\rceil\log n)\).
\end{theorem}
\begin{proof}
    Run the phase-by-phase \local implementation from
    \cref{lem:phase-local-implementation}.
    By \cref{lem:progress-lemma}, only
    \(O(\lceil \ell/\gamma\rceil\log n)\) phases are needed, and each phase
    creates at most one rake level and one compress level.  By the correctness
    lemma, the produced levels form a \((\gamma,\ell)\)-decomposition.  Since
    each phase costs \(O(\gamma+\ell^2+\ell\log^* n)\) rounds, the total
    running time is
    \(O((\gamma+\ell^2+\ell\log^* n)\lceil \ell/\gamma\rceil\log n)\).
\end{proof}

A key property of any \((\gamma,\ell)\)-decomposition is the following lemma, restated here for reader's convenience:
\lemmadist*

\begin{proof}
    The lemma holds by induction:
    Let \(V_i\) be the layer containing \(C\).
    Then clearly \(V_i \cap C^+\) contains only on component, \(C\).
    If \(V_i\) is a rake layer, then it has at most one neighboring component \(C'\) at a higher layer, in which case we can use induction on \(C'\).
    On the other hand, if \(V_i\) is a compress layer, then it has at most two neighboring components \(C'\) and \(C''\) in the higher layer.
    However, as \(V_i\) is a path with at least \(\ell\) nodes, then components \(C'\) and \(C''\) have distance at least \(\ell+1\).
    Hence, for components of \(C'\) and \(C''\), we can inductively apply the lemma for them separately.
    Finally, there cannot be a subcomponent in \(C'\) and in \(C''\) at any level such that they would be closer than \(\ell+1\) from each other as path \(V_i\) separates them.
\end{proof}

\end{document}